\documentclass[11pt]{article}
\usepackage[english]{babel}
\usepackage[latin1]{inputenc}
\usepackage{lmodern}
\usepackage[T1]{fontenc}
\usepackage[outer=3.75cm, marginparwidth=4.5cm]{geometry}
\usepackage{amssymb,amsmath,amsthm,latexsym,amsfonts,amscd,dsfont,enumerate,bbm,mathabx,mathtools}
\usepackage{graphicx,tikz}
\usepackage{lipsum}
\usepackage{enumerate} 
\usepackage{soul}
\usepackage{comment}
\usepackage{bm}
\usepackage[font=footnotesize,labelfont=bf]{caption}

\definecolor{mylightyellow}{rgb}{1,1,.8}
\definecolor{mylightgreen}{rgb}{.8,1,.8}
\definecolor{mydarkred}{RGB}{178,34,34}
\definecolor{mydarkgreen}{RGB}{34,139,34}
\definecolor{mydarkblue}{RGB}{72,61,139}
\definecolor{mydarkyellow}{RGB}{218,165,32}

\definecolor{myblueA}{RGB}{52,41,39}
\definecolor{myblueB}{RGB}{92,81,109}
\definecolor{myblueC}{RGB}{132,121,179}
\definecolor{myblueD}{RGB}{172,161,249}
\definecolor{myblueE}{RGB}{212,201,255}

\usetikzlibrary{arrows,patterns,shapes,snakes,shadows,fit,backgrounds,positioning,decorations.pathmorphing,external}
\tikzstyle{nameusd} = [circle, draw, top color=white, bottom color=mydarkblue!50, draw=mydarkblue!75!black!100, drop shadow, minimum height=4em]
\tikzstyle{nameeur} = [circle, draw, top color=white, bottom color=mydarkred!50, draw=mydarkred!75!black!100, drop shadow, minimum height=4em]
\tikzstyle{collusd} = [rectangle,fill=mydarkblue!10, inner sep=0.2cm, rounded corners=5mm]
\tikzstyle{colleur} = [rectangle,fill=mydarkred!10, inner sep=0.2cm, rounded corners=5mm]
\tikzstyle{fixto} = [draw, -latex']
\tikzstyle{fixfrom} = [draw, latex'-]
\tikzstyle{floatto} = [draw, snake=coil, segment aspect=0, line before snake=2ex, line after snake=1ex, -latex']
\tikzstyle{floatfrom} = [draw, snake=coil, segment aspect=0, line before snake=2ex, line after snake=1ex, latex'-]
\tikzstyle{investor} = [rectangle, draw, top color=white, bottom color=mydarkyellow!50, draw=mydarkyellow!75!black!100, drop shadow, rounded corners, minimum height=3em, text width=4em, text centered]
\tikzstyle{trader} = [circle, draw, top color=white, bottom color=blue!30, draw=blue!50!black!100, drop shadow, minimum height=4em]
\tikzstyle{bank} = [rectangle, draw, top color=white, bottom color=red!20, draw=red!50!black!100, drop shadow, rounded corners, minimum height=3em, text width=4em, text centered]
\tikzstyle{market} = [rectangle, draw, top color=white, bottom color=green!20, draw=green!50!black!100, drop shadow, rounded corners, minimum height=3em, text width=4em, text centered]
\tikzstyle{yellowbox} = [rectangle, draw, top color=white, bottom color=mydarkyellow!50, draw=mydarkyellow!75!black!100, drop shadow, rounded corners, text centered]
\tikzstyle{fadedyellowbox} = [rectangle, draw=gray, text=gray, top color=white, bottom color=mydarkyellow!25, draw=mydarkyellow!50, drop shadow, rounded corners, text centered]
\tikzstyle{redbox} = [rectangle, draw, top color=white, bottom color=mydarkred!50, draw=mydarkred!75!black!100, drop shadow, rounded corners, text centered]
\tikzstyle{fadedredbox} = [rectangle, draw=gray, text=gray, top color=white, bottom color=mydarkred!25, draw=mydarkred!50, drop shadow, rounded corners, text centered]
\tikzstyle{bluebox} = [rectangle, draw, top color=white, bottom color=mydarkblue!50, draw=mydarkblue!75!black!100, drop shadow, rounded corners, text centered]
\tikzstyle{fadedbluebox} = [rectangle, draw=gray, text=gray, top color=white, bottom color=mydarkblue!25, draw=mydarkblue!50, drop shadow, rounded corners, text centered]
\tikzstyle{greenbox} = [rectangle, draw, top color=white, bottom color=mydarkgreen!50, draw=mydarkgreen!75!black!100, drop shadow, rounded corners, text centered]
\tikzstyle{fadedgreenbox} = [rectangle, draw=gray, text=gray, top color=white, bottom color=mydarkgreen!25, draw=mydarkgreen!50, drop shadow, rounded corners, text centered]

\usepackage[]{hyperref}
\hypersetup{
	breaklinks, 
	baseurl = http://,
	pdfborder = 0 0 0,
	pdfpagemode = UseNone, 
	pdfstartpage = 1,
	colorlinks=true,breaklinks=true,
	urlcolor= mydarkblue,linkcolor= mydarkblue,citecolor= mydarkgreen,
	pdfauthor={Daluiso, R., Nastasi, E., Pallavicini, A., Polo, S.},
	pdftitle = {Reinforcement Learning for Options on Target Volatility Funds},
	pdfsubject = {Reinforcement Learning for Options on Target Volatility Funds}
	pdfkeywords = {Reinforcement Learning, Hedging Costs, Funding Costs, Proximal Policy Optimization, Target Volatility, Asset Allocation, Stochastic Control Problem}
}

\usepackage[backend=biber,style=authoryear-comp,maxcitenames=2,maxbibnames=99,sorting=nyt,sortcites=false, natbib=true]{biblatex}
\AtEveryBibitem{\clearfield{month}}
\AtEveryCitekey{\clearfield{month}}
\newtheorem{theorem}{Theorem}[section]
\newtheorem{lemma}[theorem]{Lemma}
\newtheorem{proposition}[theorem]{Proposition}

\newcommand{\tr}[1]{{#1}^{\intercal}} 
\DeclareMathOperator{\Tr}{Tr}
\let\norm\undefined 
\DeclarePairedDelimiter\norm{\lVert}{\rVert}

\newcommand{\virgolette}[1]{``#1''}

\newcommand{\rom}[1]{\uppercase\expandafter{\romannumeral #1\relax}}

\usepackage{eurosym}

\usepackage{cleveref}
\DeclareMathOperator*{\argmax}{arg\,max}
\DeclareMathOperator*{\argmin}{arg\,min}

\newenvironment{eqsys}{\begin{equation}\begin{dcases}}{\end{dcases}\end{equation}}
\usepackage{fancyhdr}

\date{}

\title{Reinforcement Learning for Options on \\Target Volatility Funds}

\author{
	Roberto Daluiso\thanks{Intesa Sanpaolo Group Milano, \texttt{roberto.daluiso@intesasanpaolo.com}}
	\and 
	Emanuele Nastasi\thanks{Marketz S.p.A. Milano, \texttt{emanuele.nastasi@marketz.eu}}
	\and Andrea Pallavicini\thanks{Intesa Sanpaolo Group Milano, \texttt{andrea.pallavicini@intesasanpaolo.com}} 
	\and
	Stefano Polo\thanks{X-Numeris S.r.l. Milano, \texttt{stefano.polo@xnumeris.com}} 
}

\begin{document} 

\maketitle

\begin{abstract}
\footnotesize	
\noindent
In this work we deal with the funding costs rising from hedging the risky securities underlying a target volatility strategy (TVS), a portfolio of risky assets and a risk-free one dynamically rebalanced in order to keep the realized volatility of the portfolio on a certain level. The uncertainty in the TVS risky portfolio composition along with the difference in hedging costs for each component requires to solve a control problem to evaluate the option prices. We derive an analytical solution of the problem in the Black and Scholes (BS) scenario. Then we use Reinforcement Learning (RL) techniques to determine the fund composition leading to the most conservative price under the local volatility (LV) model, for which an a priori solution is not available. We show how the performances of the RL agents are compatible with those obtained by applying path-wise the BS analytical strategy to the TVS dynamics, which therefore appears competitive also in the LV scenario.
\end{abstract}

\textbf{AMS classification codes:} 65C05, 68T07, 91G20. 
\newpage
\tableofcontents
 \mbox{}
 \vfill
 \noindent \small The opinions here expressed are solely those of the authors and do not represent in any way those of their employers.

\newpage
\pagestyle{myheadings} \markboth{}{{\footnotesize Daluiso, Pallavicini, Nastasi, Polo, Reinforcement Learning for Options on Target Volatility Funds}}

\section{Introduction}
In the recent years portfolio managers are exposed to very low interest rates and quickly changing market volatilities. An effective solution to control risks under such an environment is given by target volatility strategies (TVSs) (also known as constant volatility targeting) which are able to preserve the portfolio at a predetermined level of volatility. A TVS is a portfolio of risky assets (typically  equities) and a risk-free asset dynamically re-balanced with the aim of maintaining the overall portfolio volatility level closed to some target value. The constant volatility approach can help investors to obtain desired risk exposures over the short and long term and it increases the risk-adjusted performance of the portfolio.

These products were initially offered in the Asian markets, see for instance the reports of \textcite{Chew2011} and \textcite{Xue2012}, which highlight the pros and cons for investors, to be adopted in the following years in many other markets in North America and Europe as depicted in \textcite{Morrison2013}. At the present day we can observe some new market indices based on the mechanism of the target volatility strategy such as Dow Jones Volatility Control Index, and S\&P 500 Risk Control Index.

In the recent literature TVSs are tested to investigate their performances in term of realized returns, see for instance \textcite{Hocquard2013} and \textcite{Perchet2016}, and the soundness of the volatility targeting algorithm, as described in \textcite{Kim2018}. Moreover, in the present day, in the derivative market new structured financial products that see TVSs as underlying are becoming widely used by practitioners. Those derivative instruments are known as target volatility options (TVOs) and have gained great interest in the derivative pricing literature, see in particular \textcite{DiGraziano2012}, \textcite{Grasselli2016}, \textcite{Albeverio2019} ad \textcite{DiPersio2019}. 

This paper is the first attempt, to the best of our knowledge, to analyze the funding costs coming from hedging the risky assets underlying the target volatility strategy. In particular, we consider the point of view of a bank selling a call option to a portfolio manager as protection on the capital invested in a TVS. The portfolio manager has the freedom of changing the relative weights of the risky assets during the life of the TVS. Since the risky assets have different hedging costs, the bank shall adjust the price of the protection to include them in the worst-case scenario, i.e. the most expensive strategy in term of the financing costs. Hence, the pricing problem becomes a continuous dynamical control problem over the risky portfolio composition. In our work we provide to the reader a formal description of the stochastic control problem and most importantly we derive an analytical solution assuming that the risky asset dynamics underlying the TVS  follow a Black and Scholes (BS) model \citep{BlackScholes1953} and that the derivative contract is a European-style option. 

Although from a theoretical point of view this solution represents a sub-optimal strategy when dealing with generic dynamics for the risky assets, we show that the BS solution provides a strong baseline as price to practitioners when testing it against numerical methods with different levels of sophistication for solving the control problem. More precisely we numerically study the problem in the general case of a local volatility dynamics \citep{Dupire1994,Derman1994} by adopting a Reinforcement Learning approach, in particular an \textit{ad hoc} vanilla direct policy algorithm and the more sophisticated proximal policy optimization technique developed in \textcite{Schulman2017}. 
\\

The paper is organized as follows. In \Cref{sec:TVS} we introduce more in detail the manager-bank contract and provide the description of the TVS dynamics in presence of valuation adjustments such as the hedging costs. Then, in \Cref{sec:Derivative} we introduce the structured class of derivative contracts linked to TVSs, where we describe the arising dynamical control problem for pricing those options. Moreover in this section we derive the BS closed solution for European TVOs in two different ways: one applying the Gy\"ongy Lemma and the other by writing the Hamilton-Jacobi-Bellman equation. In \Cref{sec:RL} we illustrate how we have applied RL to solve the dynamic control problem, giving a description of the algorithm we have built. We conclude the paper with \Cref{sec:Numerical_investigations} where we present the numerical results obtained in this work for the Black and Scholes model and the local volatility one.
\\

A part of this work has been developed during the Master thesis of one of us (SP), where we had the opportunity to collaborate with Marco Bianchetti and Diego Pierluigi Giovannini from Intesa Sanpaolo Milan, that we tank. Moreover we wish to thank the Italian computing centre Cineca, which approved our ISCRA C project and provided us with the high-computing resources of Marconi100 for the numerical simulations of this contribution.

\section{Target Volatility Strategy}\label{sec:TVS}
As mentioned before, this work aims to enrich the TVS pricing literature by studying the aspect related to the funding costs coming from hedging the risky assets underlying a target volatility portfolio. Thus we consider the following scenario: a bank selling a protection to a portfolio manager who has his capital invested in a TVS. In our case the fund manager has the freedom of changing the relative weights of the risky asset during the life of the TVS; once the allocation strategy is selected then the volatility targeting algorithm rebalances the risky component of the portfolio with the risk-free one in order to keep the overall portfolio volatility close to a target value. Clients investing in the fund pay a running fee for the service of the fund manager and their capital is protected. 

The fund manager usually buys from a bank an option on the TVS to ensure capital protection. For instance, the capital can be protected by buying a put option. In this case, we can write the net asset value (NAV) $A_t$ of the strategy as given by\footnote{Here we neglect discounting factors.}
\begin{equation}
    A_t \coloneqq \max\{V_t,K\} = V_t + (K-V_t)^+,
\end{equation}
where $V_t$ is the price process of the strategy, and $K$ is the guaranteed capital. On the other hand, the fund manager can replicate the payoff by means of the put-call parity by investing the capital in a low-risk asset and buying a call on the strategy
\begin{equation}
    A_t = K + (V_t-K)^+.
\end{equation}
In this way, the TVS is only defined in the two contracts client-fund and fund-bank. The fund manager is not implementing the strategy by trading in the market, and he is not subject to additional costs to access the market. On the other way, the bank is paying such costs since she is actively hedging the call option sold to the manager.  The bank trading activity implemented to actively hedge the option requires funding the collateral procedures of the hedging instruments along with any lending/borrowing fee.
We remark that the choices of the manager trading activity are stochastic processes since they will depend on the market evolution. Thus, neither the manager strategy nor the bank financing costs for hedging the option can be written in the fund-bank contract. For this reason, the price of a financial product sold by the bank is adjusted to include any valuation adjustment due to the trading activity.

Our aim is to find the most expensive investing strategy from the point of view of financing costs for the bank that the manager could choose in the market. In other words we want to determine the worst-case scenario for the bank. In this section we proceed by defining the price process of the TVS so that we can highlight the impact of valuation adjustments.

\subsection{The strategy Price Process}
We work on a filtered probability space $\left(\Omega, \mathcal{F}, \{\mathcal{F}_t\}_{t\geq0}, \mathbb{P}  \right)$ satisfying the usual assumptions for a market model, where $\mathbb{P}$ is the physical probability measure representing the actual distribution of supply and demand shocks on equities prices. 

We consider a fund trading a basket of $n$ risky securities with price process $S_t = (S_t^i, \, i=1,\dots,n)$ with $t\in I \subset \mathbb{R}$ funded with a cash account $B_t$ accruing at $r_t$. Any dividend paid by the securities is re-invested in the fund. Here, we assume that the TVS is implemented in continuous time, even if in the practice we can implement the strategy only on a discrete set of dates. We introduce the deflated gain process $\Bar{G}_t^i$ associated with the risky securities as given by
\begin{equation}
    \Bar{G}_t^i \coloneqq \Bar{S}_t^i + \Bar{D}_t^i,
\end{equation}
where we define the deflated price and cumulative dividend processes as 
\begin{equation}
    \bar{S}_t^i \coloneqq \frac{S_t^i}{B_t}, \quad \bar{D}^i_t \coloneqq \int_0^t\frac{d\pi_u^i}{B_u} + \int_0^t\frac{d\psi_u^i}{B_u},
\end{equation}
where $\pi^i_t$ represents the cumulative contractual-coupon process paid by the security, and $\psi_t^i$ represents the cumulative valuation adjustments.
Since fund managers allocating TVS usually rely on equity assets, here we use the results of \textcite{Gabrielli2020} who analyze the valuation adjustments for equity products. We can write 
\begin{equation}
    \psi_t^i \coloneqq \int_0^t S_u^i\mu_u^i du,
\label{eq:XVA_equity}\end{equation}
where we call $\mu_t^i$ cost of carry, which basically represents the hedging costs for the $i$-th security.

Then, we introduce the strategy price process $V_t$, and we define the deflated gain process $\Bar{G}_t^V$ as given by
\begin{equation}
    \Bar{G}_t^V \coloneqq \frac{V_t}{B_t} + \int_0^t\frac{V_u\phi_u}{B_u}du,
\end{equation}
where $\phi_t$ are the running fees earned by the fund manager for his activity. We assume that the strategy is self-financing, so that we can write
\begin{equation}
    d\Bar{G}_t^V = q_t \cdot d\bar{G}_t,
\label{eq:self_financing}\end{equation}
where $q_t^i$ is the quantity invested in the $i$-th security\footnote{In all formulae we use dot notation for scalar product between vectors, i.e. $a \cdot b = \sum_i a_i b_i$, or between matrix and vector, i.e. $A \cdot b = \sum_j a_{ij}b_j$ or $b \cdot A = \sum_i b_i a_{ij}$.}.

Now, in order to prevent arbitrages, we assume the existence of a risk-neutral measure $\mathbb{Q}$ equivalent to $\mathbb{P}$ under which the deflated gain processes of all traded securities are martingales. Under this assumption, we are able to derive the drift conditions on the security price processes, and in turn on the strategy price process.
\begin{equation}
    \forall u >t \quad \bar{G}_t^i = \mathbb{E}_t \left[\bar{G}_u^i\right] \quad\Longrightarrow \quad dS_t^i = r_tS_t^idt-d\pi_t^i-d\psi_t^i + dM_t^i,
\label{eq:risk_neutral}\end{equation}
where $M_t^i$ are martingale under $\mathbb{Q}$. If we substitute this expression for the security dynamics into the definition of the strategy we can check that the price process of the strategy is accruing at a cash account rate $r_t$ compensated for the fund manager fees
\begin{equation}
    dV_t = V_t(r_t-\phi_t)dt + dM_t^V,
\end{equation}
with $M_t^V$ martingale under $\mathbb{Q}$. Notice that, as expected from non-arbitrage considerations, the coupons paid by each security appear only in the drift of the security price process, but they do not impact the drift of the strategy. 

Yet, the strategy priced by $V_t$ cannot be described in the contract between the parties, since \Cref{eq:self_financing} depends via the security gain processes on the valuation adjustment $\psi_t^i$, which is specific to the investor pricing the strategy. Thus, the TVS defined in the contract will be
\begin{equation}
    d\bar{I}_t \coloneqq q_t \cdot \left(d\bar{S}_t + \frac{d\pi_t}{B_t} \right) - \bar{I}_t\phi_t dt \quad \text{with }I_0 = V_0, 
\label{eq:TVS_first}\end{equation}
leading to the following price process dynamics
\begin{equation}
    dI_t = I_t(r_t-\phi_t)dt -q_t\cdot d\psi_t +dM_t^I,
\end{equation}
with $M_t^I$ martingale under $\mathbb{Q}$. In this case we observe that $I_t$ depends explicitly both on the valuation adjustments and on the allocation strategy. Indeed, if we substitute the valuation adjustments with their explicit expression (\Cref{eq:XVA_equity}), we get
\begin{equation}
    dI_t = I_t(r_t-\phi_t)dt - q_t \cdot S_t \mu_t dt + dM_t^I,
\end{equation}
where we can see the dependency on cost of carry $\mu_t^i$ and on the allocation strategy.

\subsection{The Volatility Targeting Constraint}
In a typical TVS, the fund manager selects a risky-asset portfolio with a specific time-dependent allocation strategy expressed by means of the vector of relative weights $\alpha_t$, along with a risk-free asset, which we can identify with the bank account $B_t$. In the following, for sake of simplicity in exposition we consider only total-return securities, namely we set 8$\pi_t=0$. Thus we can write \Cref{eq:TVS_first} as given by
\begin{equation}
    \frac{dI_t}{I_t}= \omega_t \alpha_t \cdot \frac{dS_t}{S_t} + \left(1-\omega_t  \alpha_t \cdot \mathds{1} \right)\frac{dB_t}{B_t}- \phi_t dt,
\label{eq:TVS_elegant}\end{equation}
where $\mathds{1}$ is a $n$-dimensional vector of ones and $\omega_t\in[0,1]$ is determined so that the strategy log-normal volatility is kept constant, namely
\begin{equation}
        \omega_t: \quad \mathrm{Var}_t[dI_t] = \bar{\sigma}^2I_t^2dt,
\end{equation}
where $\Bar{\sigma}$ is the target volatility value.
In practice, this means that the fund manager will select a risky-portfolio choosing $\alpha_t$ equities from the universe where he can trade and after that, his choices will be scaled by the automatic target volatility algorithm\footnote{We recall that the universe of assets where the manager can trade and the value of $\bar{\sigma}$ are written in the contract.} $\omega_t$.

To derive the expression for $\omega_t$ we need to assume a generic continuous semi-martingales dynamics under the risk-neutral measure for the underlying securities, so that we can write \Cref{eq:risk_neutral} as 
\begin{equation}
    \frac{dS_t^i}{S_t^i} = \left(r_t - \mu_t^i \right)dt + \nu_t^i \cdot dW_t,
\label{eq:Equity_process}\end{equation}
where $\nu_t$ is an adapted matrix process ensuring the existence of a solution for the stochastic differential equation (SDE) and $W_t$ is a $n$-dimensional vector of Brownian motions under $\mathbb{Q}$. Under these assumptions we can derive an expression for $\omega_t$, and we get\footnote{In all formulae the norm for a vector $a$ is defined as $\|a\|\coloneqq\sqrt{a\cdot a}$.}
\begin{equation}
    \omega_t = \frac{\Bar{\sigma}}{\|\alpha_t\cdot \nu_t \|}.
\end{equation}
Hence, putting this last result in the dynamics of $I_t$ we obtain

\begin{equation}
     \frac{dI_t}{I_t} = \left(r_t -  \phi_t     - \frac{\bar{\sigma} \alpha_t}{\|\alpha_t \cdot \nu_t \|}  \cdot \mu_t \right)dt+ \frac{\bar{\sigma}\alpha_t }{\|\alpha_t \cdot \nu_t \|} \cdot \nu_t \cdot dW_t,
\label{eq:TVS_last}\end{equation}
where we can see, as expected, that the strategy grows at the risk-free rate but for adjustments due to valuation adjustments and fees.

We highlight that, to derive the dynamics expressed in \Cref{eq:TVS_last}, we have not made any assumptions on the risky allocation strategy; thus all this argument is valid for any constraints on the process $\alpha_t$.

\section{Derivative Pricing}\label{sec:Derivative}
In this section, we analyze derivative contracts linked to the TVS described previously. In particular we will focus on European style options, and we will show that under appropriate assumptions it is possible to find a closed form solution for the optimal allocation strategy which maximizes the contract price.

In a general framework, a derivative contract on the TVS with maturity $T$ can be defined as
\begin{equation}
    V_0 \coloneqq \sup_\alpha \mathbb{E}_0\left[\int_0^T D(0,u;\zeta_u)d\pi_u(\alpha)\right],
\end{equation}
where $D(0,T;\zeta_t)$ is the discount factor with rate $\zeta_t$, inclusive of the derivative valuation adjustments, and $\pi_t$ is the cumulative coupon process paid by the derivative, and it depends on the allocation strategy since in turn the TVS depends on it via the valuation adjustments. We take the supremum over the strategies since we do not have any information on the future activity of the fund manager.

\subsection{European Options}
If the derivative contract depends only on the marginal distribution of $I_T$ at maturity (a European payoff), we are able to prove that exists an optimal strategy, and we are able to calculate it. We consider the following pricing problem
\begin{equation}
    V_0 \coloneqq \sup_\alpha \mathbb{E}_0\left[D(0,T;\zeta)\Phi(I_T(\alpha))\right],
\end{equation}
where $\Phi$ is the payoff function of the derivative. 

We start by introducing the Markovian projection of the dynamics followed by $I_t$. We name it $I_t^{\text{MP}}$, and we get by applying the Gy\"ongy Lemma \citep{Gyongy1986}
\begin{equation}
    \frac{dI_t^{\text{MP}}}{I_t^{\text{MP}}} \coloneqq \left(r_t - \ell_{\alpha} \left(t,I_t^{\text{MP}}\right)\right)dt + \bar{\sigma}dW_t^{\text{MP}} \quad \text{with} \quad I_0^\text{MP} = I_0,
\label{eq:markovian_projection}\end{equation}
where the local drift is defined as
\begin{equation}
    \ell_\alpha \left(t,K\right) \coloneqq \bar{\sigma} \mathbb{E}_0\left[\frac{\mu_t \cdot \alpha_t}{\|\alpha_t \cdot \nu_t \|}\bigg|I_t = K\right],
\label{eq:local_drift}\end{equation}
and $W_t^{\text{MP}}$ is a Brownian motion under the risk-neutral measure $\mathbb{Q}$. Notice that the diffusion coefficient collapses to the target volatility value $\bar{\sigma}$. Since European payoffs depend only on the marginal distribution at maturity, they can be calculated only by means of the Markovian projection $I_t^{\text{MP}}$, namely
\begin{equation}
V_0 \coloneqq \sup_{\alpha} \mathbb{E}_0\left[D\left(0,T;\zeta\right)\Phi\left(I_T^{\text{MP}}\left(\alpha\right)\right)\right].
\end{equation}
Hence, we have our first result valid only if valuation adjustments can be neglected:
\begin{proposition}
A European payoff on the TVS can be calculated by assuming any allocation in the underlying risky basket if all the underlying
securities grow under the risk-neutral measure at the risk-free rate without any valuation adjustment, namely if we can write $\mu_t=0$.
\end{proposition}

\subsection{Stochastic Optimal Control Problem}
In presence of valuation adjustments, we need to solve the full optimization problem. We discretize the optimal strategy $\alpha_t$ as\footnote{We use the symbol $\mathbf{1}_A$ for the indicator function of a subset $A$.}
\begin{equation}
    \alpha_t \coloneqq \sum_k \mathbf{1}_{ \{t \in [T_{k-1}, T_k)\}}\alpha_{T_{k-1}},
\label{eq:piecewise_strategy}\end{equation}
according to a time grid $\mathcal{T}\coloneqq \{T_0,...,T_k,...,T_m\}$ with $T_0\coloneqq t$ the pricing date and $T_m \coloneqq T$ the maturity of the option. 
Therefore we can apply the dynamic programming principle to express the optimal $\alpha_t$ at time $T_{k-1}$ as
\begin{equation}
    \alpha_{T_{k-1}}\coloneqq \argmax_{\alpha} \left\{\mathbb{E}_{T_{k-1}}\left[D\left(T_{k-1},T_{k} \right) V_{T_{k}}\left(X_{T_{k}}, I_{T_{k}}(\alpha)\right) \mid X_{T_{k-1}}, I_{T_{k-1}}\right]\right\},
\label{eq:recursion}\end{equation}
where $V_{T_k}$ is the option value at time $T_k$ and $X$ is any Markovian state such that the drift and the diffusion coefficient of $I_t$ are a function of $\left(X_t,I_t,\alpha_t\right)$. We calculate $I_{T_k}\left(\alpha_{T_{k-1}}\right)$ for any given strategy $\alpha_{T_{k-1}}$ by a suitable discretization of \eqref{eq:TVS_last} starting from $X_{T_{k-1}}$ and $I_{T_{k-1}}$.

Thus the derivative price is given by:
\begin{equation}
    V_{T_{k-1}}\left(X_{T_{k-1}}, I_{T_{k-1}}\right)=\mathbb{E}_{T_{k-1}}\left[D\left(T_{k-1},T_{k}\right) V_{T_{k}}\left(X_{T_{k}}, I_{T_{k}}\left(\alpha_{T_{k-1}}\right)\right) \mid X_{T_{k-1}}, I_{T_{k-1}}\right],
\label{eq:recursion2}\end{equation}
while the iteration starts from maturity date where the boundary condition is set equal to the payoff function:
\begin{equation}
    V_{T_m} = \Phi\left(I_{T_m}\right).
\end{equation}

\subsection{Black and Scholes Model}\label{subsec:BS_model}
In the time-dependent Black and Scholes model with deterministic rates, we can work with empty $X_t$, since in this case the portfolio dynamics \eqref{eq:TVS_last} is Markovian, leading to an optimal strategy $\alpha_t^*$ which depends in principle only on $I_t$. As a consequence, the local drift defined in \Cref{eq:local_drift} can be written as
\begin{equation}
     \ell_\alpha \left(t,K\right) = \bar{\sigma}   \frac{\mu\left(t\right) \cdot \alpha\left(t,K\right)}{\|\alpha\left(t,K\right) \cdot \nu\left(t\right)\|},  
\end{equation}
so that the optimization problem can be solved by looking only at the Markovian projection without simulating all the Brownian motions $W_t$. Notice that we are indicating the dependency on time in parenthesis to highlight that in this formula all the quantities are deterministic functions of time.

A direct consequence is the following proposition, which is relevant for plain vanilla options on TVS.

\begin{proposition}
When the underlying securities follow a Black and Scholes model with deterministic rates, the optimal strategy for a non-decreasing European payoff consists in minimizing the local drift function, independently of the current state $I_t$
\begin{equation}
    \alpha^*(t) \coloneqq \argmin_\alpha \frac{\alpha \cdot \mu(t)}{\|\alpha \cdot \nu(t) \|}.
\label{eq:BS_optimal_strategy}
\end{equation}
Analogously, the optimal strategy for a non-increasing European payoff consists in maximizing the local drift function:
\begin{equation}
    \alpha^*(t) \coloneqq \argmax_\alpha \frac{\alpha \cdot \mu(t)}{\|\alpha \cdot \nu(t) \|}.
\label{eq:BS_optimal_strategy_put}
\end{equation}
\end{proposition}
The absence of stochastic elements in \Cref{eq:BS_optimal_strategy} makes the optimal strategy known \textit{a priori} with no simulation needed; in fact one can solve the optimization problem once for all $\forall t \in \mathcal{T}$ just looking at the market data $\mu(t)$ and $\nu(t)$ for the securities. Once $\alpha^*$ is known, then one can price the payoff by the following BS formula
\begin{equation}
	V_0^{\text{BS}} = BS(F^{\text{TVS}}(0,T;\alpha^*), K, T, \bar{\sigma}, D(0,T;\zeta)),
\label{eq:bs_closed_solution}\end{equation}
where $F^{\text{TVS}}(t,T;\alpha)$ is the TVS forward curve defined by
\begin{equation}
	F^{\text{TVS}}(t,T;\alpha) = I_t \exp\left[\int_t^T\left(r(u)-\ell_{\alpha} \left(u\right)\right)du\right],
\label{eq:tvs_forward}\end{equation}
while $BS(F,K,T,\sigma,D)$ is the standard BS formula for a European option with forward curve $F$, strike $K$, time to maturity $T$, volatility $\sigma$ and discount factor $D$. 

In \Cref{fig:strategy_price_comparison} we provide a comparison among the option price obtained with the optimal stragety (BS$^*$) of \Cref{eq:BS_optimal_strategy} and those with other intutive strategies (S$_A$, S$_B$, S$_C$). Here we consider the case of an at-the-money call option with spot $I_0 = 1\, \text{EUR}$, maturity $T=5 \, \text{yr}$, target volatility $\bar{\sigma}=5\%$, and a nonnegative constraint on $\alpha$. The intuitive strategies are
\begin{itemize}
\item S$_A$: invest all in the asset with the maximum forward curve at maturity;
\item S$_B$: for each market pillar before maturity invest all in the asset with minimum $\mu(t)$; 
\item S$_C$: for each market pillar before maturity invest all in the asset with minimum $\mu(t)/\|\nu(t)\|$.
\end{itemize}
As the reader can observe, our strategy outperforms any other intuitive one a practitioner might adopt.
\begin{figure}[t]
	\centering
	\includegraphics[height=5cm,width=10cm]{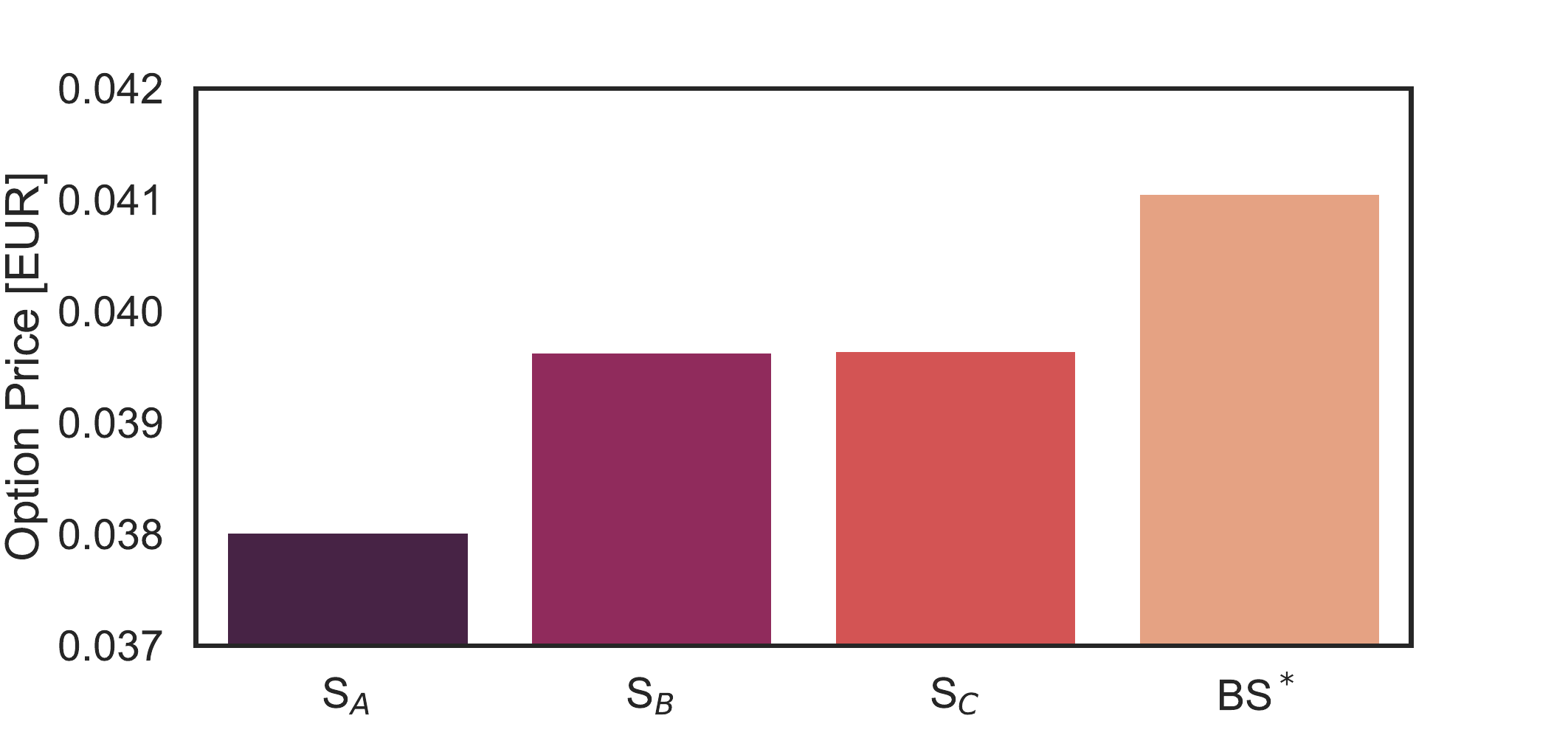}
	\caption{Comparison of plain vanilla prices on the TVS by adopting different allocation strategies: the optimal Black and Scholes solution (BS$^*$) of \Cref{eq:BS_optimal_strategy} and other intuitive strategies (S$_A$, S$_B$, S$_C$).}
	\label{fig:strategy_price_comparison}
\end{figure}

\subsubsection{Unconstrained Allocation Strategy: Closed Form Solution}
In absence of constraints on the allocation strategy, we are able to derive a closed form solution to the BS problem \eqref{eq:BS_optimal_strategy}.
\begin{lemma}
	Let be $\mu, \alpha \in \mathbb{R}^n$, $\nu \in \mathbb{R}^{n\times n}$ be a full rank matrix and $\Sigma\coloneqq \nu \tr{\nu}$. Then the closed solution of the optimization problem \eqref{eq:BS_optimal_strategy} is 
	\begin{equation}
			\alpha^* = - \frac{\Sigma^{-1} \cdot \mu}{\|(\Sigma^{-1} \cdot \mu)\cdot \nu\|}.
	\end{equation}
\end{lemma}
\begin{proof}
Since the argument of the minimum \eqref{eq:BS_optimal_strategy} is zero-homogeneous, then we can rewrite the problem as
\begin{equation}
		\begin{aligned}
			&\text { minimize } \alpha \cdot \mu\\
			&\text { subject to } \|\alpha \cdot \nu\|^2=1
	\end{aligned}
\end{equation}
By setting the Lagrangian function associated with the problem
\begin{equation}
	\mathcal{L}\left(\alpha, \lambda\right)=\alpha \cdot \mu-\lambda\left(\|\alpha \cdot \nu\|^2-1\right) \, ,
\end{equation}
we obtain the first order conditions
\begin{eqsys}
	\frac{\partial \mathcal{L}}{\partial \alpha}=\mu -2\lambda \Sigma \cdot \alpha=0 \\
	\frac{\partial \mathcal{L}}{\partial \lambda}= \|\alpha \cdot \nu\|^2-1=0
	\, ,
\end{eqsys}
Then, by applying simple algebra, we obtain the analytical form for the free optimal strategy
\begin{equation}
	\alpha^* = \pm \frac{\Sigma^{-1} \cdot \mu}{\|(\Sigma^{-1} \cdot \mu)\cdot \nu\|}.
\label{eq:analytical_solution_free_strategy}\end{equation}
We take the minus sign to get the minimum value of the TVS local drift while the plus sign for the maximum one (put payoff). 
\end{proof}

\subsubsection{Active Asset (or Bang Bang) Solution}
A closed form solution to the minimization of the local drift correction \eqref{eq:BS_optimal_strategy} can also be derived in the common case that all costs of carry are nonnegative and the only constraint on portfolio weights is nonnegativity, which would mean a long-only strategy by the fund manager.

\begin{lemma}
	Let $\mu \in \mathbb{R}^n$ be a vector with nonnegative components, $\nu \in \mathbb{R}^{n\times n}$ be a full rank matrix, and $\Sigma = \nu \tr{\nu}$. Then
	\begin{equation}\label{eq:inf}
		\inf_{\alpha \in \mathbb{R}_{+}^n \setminus \{0\}} \frac{\alpha \cdot \mu}{\norm{\alpha \cdot \nu}} = \min_{i\leq n} \frac{\mu_i}{\sqrt{\Sigma_{ii}}};
	\end{equation}
	if $\bar{\imath}$ is the index which realizes the min, then the infimimum is realized by a vector concentrated on the $\bar{\imath}$ component: $\alpha_i = \delta_{i\bar{\imath}}$. 
\end{lemma}
\begin{proof}
	Let us first consider the case in which $\mu = \mathbbm{1}$. Since the argument of the infimum is zero-homogeneous, normalizing by $\alpha \cdot \mathbbm{1} > 0$ we can restrict to the affine hyperspace $\{\alpha \cdot \mathbbm{1} = 1\}$, where the minimization \eqref{eq:inf} reduces to the maximization of its denominator: the required infimum will be the square root of the reciprocal of
	\[
	\sup\left\{\norm{\alpha \cdot \nu}^2 \mid \alpha \in \mathbb{R}_{+}^n,\, \alpha \cdot \mathbbm{1} = 1\right\}.
	\]
	Now we can note that $\Sigma$ is positive definite, hence $\Sigma_{ij} < \sqrt{\Sigma_{ii}\Sigma_{jj}} \leq \Sigma_{\bar{\imath}\bar{\imath}}$, which implies
	\[
	\norm{\alpha \cdot \nu}^2 = \sum_{i,j=1}^{n} \alpha_i \alpha_j \Sigma_{ij} \leq \sum_{i,j=1}^{n} \alpha_i \alpha_j \Sigma_{\bar{\imath}\bar{\imath}} = \Sigma_{\bar{\imath}\bar{\imath}}
	\]
	because $\sum_i \alpha_i = 1$. Since we trivially have equality for $\alpha_i = \delta_{i\bar{\imath}}$, this concludes the proof of the case $\mu = \mathbbm{1}$.
	
	Next, let us consider the case in which all components of $\mu$ are strictly positive, and define $M$ as the diagonal matrix with diagonal $\mu$. Then we can rewrite the infimum as a function of $\beta = M\alpha$:
	\[
	\inf_{\beta \in \mathbb{R}_{+}^n \setminus \{0\}} \frac{\beta \cdot \mathbbm{1}}{\norm{\beta \cdot M^{-1}\nu}},
	\]
	which by the first part of the proof equals 
	\[
	\min_{i \leq n} \frac{1}{\sqrt{\tilde{\Sigma}_{ii}}} = \min_{i \leq n} \frac{\mu_i}{\sqrt{\Sigma_{ii}}}, \quad \tilde{\Sigma} \coloneqq M^{-1}\nu\nu^{T} M^{-1} = M^{-1}\Sigma M^{-1}.
	\]
	
	Finally, let us consider the general case in which $\mu$ may have some components equal to zero. For an arbitrary $\epsilon \geq 0$ let us define
	\[
	f_{\epsilon}(\alpha) = \frac{\alpha \cdot (\mu+\epsilon)}{\norm{\alpha \cdot \nu}}.
	\]
	One can easily note that as $\epsilon \to 0$, $f_{\epsilon}$ tends to $f_{0}$ uniformly on the compact set $\{\alpha \in \mathbb{R}^n_+ \mid \alpha \cdot \mathbbm{1} = 1\}$, so that the minimum converges to the minimum on that set. Since we know by homogeneity that the minimum on $\{\alpha \in \mathbb{R}^n_+ \mid \alpha \cdot \mathbbm{1} = 1\}$ equals the minimum on $\mathbb{R}^n_+ \setminus \{0\}$, we conclude
	\[
	\inf_{\alpha \in \mathbb{R}_{+}^n \setminus \{0\}} f_0(\alpha) = 
	\lim_{\epsilon \to 0+} \inf_{\alpha \in \mathbb{R}_{+}^n \setminus \{0\}} f_{\epsilon}(\alpha) = 
	\lim_{\epsilon \to 0+} \min_{i\leq n} \frac{\mu_i+\epsilon}{\sqrt{\Sigma_{ii}}} =
	\min_{i\leq n} \frac{\mu_i}{\sqrt{\Sigma_{ii}}}.
	\]
\end{proof}

\subsection{Hamilton-Jacobi-Bellman Equation for Target Volatility Options}\label{subsec:HJB}
In this section, we want to provide to the reader a formal description of the dynamic problem associated with options on target volatility strategies by writing the Hamilton-Jacobi-Bellman (HJB) equation for the derivative price. We prove that from this equation one can recover the same closed formula \eqref{eq:BS_optimal_strategy} for the time-dependent BS model which was derived above from the Gy\"ongy Lemma. 

In full generality we assume that the time evolution of the Markovian factors governing the problem dynamics is given by  a stochastic multidimensional process in $\mathbb{R}^n$, $X_t$, that is the unique strong solution to the following It$\hat{\text{o}}$ SDE 
\begin{equation}
		dX_t = M(X_t)dt + \Sigma(X_t) \cdot dW_t,
\end{equation}
which is driven by a $n$-dimensional Wiener process with independent components, $W_t$. We point out to the reader that with this notation we are including general dynamics models like those with stochastic drift $M(X_t)$ and stochastic diffusive $\Sigma(X_t)$ coefficients. Let be the dynamics of the securities $S_t$ a generic function of the Markovian factors, namely
\begin{equation}
 S_t \coloneqq  f(X_t).
\label{eq:S_f}\end{equation}

In this framework, the TVS price process dynamics is given by the SDE 
\begin{equation}
	\frac{dI_t}{I_t} = \left(r(X_t) -  \frac{\bar{\sigma} \alpha_t}{\|\alpha_t \cdot \nu(X_t)\|}  \cdot \mu(X_t)    \right)dt +\frac{\bar{\sigma}\alpha_t }{\|\alpha_t \cdot \nu(X_t) \|} \cdot \nu(X_t) \cdot dW_t,
\end{equation}
where the expression of $\mu(X_t)$ and $\nu(X_t)$ can be recovered by applying the It$\hat{\text{o}}$ formula to \Cref{eq:S_f}.

Given $X \coloneqq X_t$ and $I \coloneqq I_t$, we can write the HJB equation for $V \coloneqq V(t,I,X)$ as follows
\begin{equation}
	\begin{aligned}
		\frac{\partial V}{\partial t} + \max_\alpha  \left\{   \left(r(X) - \bar{\sigma}\frac{\alpha \cdot \mu(X)}{\|\alpha \cdot \nu(X)\|}\right) I \frac{\partial V}{\partial I} + (\nabla_X V) \cdot M(X) + \frac{1}{2} \bar{\sigma}^2I^2 \frac{\partial^2V}{\partial I^2}  \right.\\ 
		\left. +\frac{1}{2}\Tr\left(\tr{\Sigma(X)}(H_XV)\Sigma(X)\right)  + (\nabla_{X,I}V) \cdot \Sigma(X) \cdot \left(I\bar{\sigma}\frac{\alpha \cdot \nu(X)}{\|\alpha \cdot \nu(X)\|}\right)\right\} = 0,
	\end{aligned}
\end{equation}
where $\Tr(A)$ is the trace operator of $A$, $\nabla_X V$ the gradient of $V$ w.r.t. $X$, $H_XV$ the Hessian matrix of $V$ w.r.t.~$X$ and $\nabla_{X,I}V$ is the vector defined by:
\begin{equation}
	\nabla_{X,I} V \coloneqq \tr{\left(\frac{\partial^2 V}{\partial X^1 \partial I},\dots, \frac{\partial^2 V}{\partial X^n \partial I} \right)}.
\end{equation}
We take out from the maximum operator all the elements that do not depend on the risky allocation strategy $\alpha$
\begin{equation}
	\begin{aligned}
		\frac{\partial V}{\partial t} +   r(X) I \frac{\partial V}{\partial I} + (\nabla_X V) \cdot M(X) + \frac{1}{2} \bar{\sigma}^2I^2 \frac{\partial^2V}{\partial I^2} + \frac{1}{2}\Tr\left(\tr{\Sigma(X)}(H_XV)\Sigma(X)\right)  \\ 
		+\bar{\sigma}I\max_\alpha  \left\{  -\frac{\partial V}{\partial I}   \frac{\alpha \cdot \mu(X)}{\|\alpha \cdot \nu(X)\|} + 
		(\nabla_{X,I}V) \cdot \Sigma(X) \cdot \left(\frac{\alpha \cdot \nu(X)}{\|\alpha \cdot \nu(X)\|}\right)\right\} = 0.
		\end{aligned}
\label{eq:HJB_general}\end{equation}
\Cref{eq:HJB_general} is the Hamilton-Jacobi-Bellman equation describing the TVS dynamic problem for a generic dynamics of the risky securities underlying the portfolio.

If we assume a time-dependent BS dynamics for the risky equities ($\mu_t$, $r_t$ and $\nu_t$ deterministic), then $V = V(t,I)$ and all the derivatives w.r.t. $X$ are zero. Therefore the reduced HJB equation is
\begin{equation}
	\frac{\partial V}{\partial t} +  r(t) I \frac{\partial V}{\partial I}  + \frac{1}{2}\bar{\sigma}^2I^2\frac{\partial^2V}{\partial I^2} + \bar{\sigma}I\max_\alpha \left\{ -\frac{\partial V}{\partial I}   \frac{\alpha \cdot \mu(t)}{\|\alpha \cdot \nu(t)\|}\right\} = 0;
\end{equation}
if the payoff is non-decreasing in $I$ by homogeneity of the SDE we get that $V$ is non-decreasing; thus the solution is given by
\begin{equation}
 \alpha^*(t) = \argmin_\alpha \frac{\alpha \cdot \mu(t)}{\|\alpha \cdot \nu(t)\|},
	\end{equation}
which is the same result expressed in \Cref{eq:BS_optimal_strategy}. On the other hand, if the payoff is non-increasing in $I$ then the solution will be the $\argmax$.

Conversely, if we deal with a dynamic model for which the derivative contract $V$ depends on $X$ then the volatility versor term, namely the second one, in the $\max$ operator of \Cref{eq:HJB_general} is no longer zero and thus one must solve the entire control problem numerically. 

In this work we tackle the non-trivial case of a local volatility (LV) model for the $S_t$-dynamics, such that $\nu_t = \nu(t,S_t)$. We have chosen the LV model since it is well known in financial literature and among practitioners.

\section{Reinforcement Learning}\label{sec:RL}
As we have discussed in the previous sections, one must resort to numerical approaches to solve the stochastic control problem related to the TVS in the case of general payoffs or risky securities dynamics. The standard approach could be to use classical techniques based on backwards recursion \eqref{eq:recursion}-\eqref{eq:recursion2} such as American Monte Carlo \citep{Longstaff2001}. However, their performances degrade exponentially as the dimension $n$ of the problem increases, making prohibitively costly finding the solution to the problem. In our contribution, we adopt a novel technique which is free from the curse of dimensionality and is gaining popularity in many scientific branches for solving stochastic optimal control problems: Reinforcement Learning (RL) \citep{Sutton2018}. 

Reinforcement Learning is a branch of Machine Learning which allows an artificial agent to interact with an environment through actions and observations in order to maximize some notion of cumulative reward. In RL the agent is not told which actions to take but instead it must discover by trial and error which are the behaviours  yielding the highest reward. This is obtained by updating the agent policy $\pi$ which is a mapping from the environment states to the set of actions. Thus RL is independent of pre-collected data as opposed to other Machine Learning techniques. Because of its nature, RL has been successful in quantitative finance for solving control problems; among the most important RL applications in this field, we refer to \textcite{Deng2017} as the pioneers in studying self-taught reinforcement trading problems, while to \textcite{Kolm2019} and \textcite{Halperin2020} for hedging derivatives with RL under market frictions.

In our work we adopt two learning strategies to compare their performances in terms of overall reward: an \textit{ad hoc} direct policy algorithm and the state of the art proximal policy optimization (PPO) developed in \textcite{Schulman2017} and \textcite{Schylman2016}. 

The first method is a specific algorithm developed by us to fit the problem aim to find the optimal option price; this technique can take place as a direct policy optimization method in the wide taxonomy of Reinforcement Learning algorithms. We will go into details in \Cref{sec:direct_policy}.

On the other hand, the PPO is a high-level actor-critic algorithm well-suited for continuous control problems. It collects a small batch of experiences interacting with the environment to update its decision-making policy. From those interactions with the environment, PPO is able to compute the expected reward and the value function. We will not provide a complete description of this sophisticated learning strategy; for more details, we refer to the authors' papers. In our work, we adopt the PPO implementation found in \textsc{OpenAI Baselines} \citep{Dhariwal2017}. 

In the following sections, we describe the way we have formalized the TVS problem in the Reinforcement Learning framework.

\subsection{Direct Policy Approach}\label{sec:direct_policy}
We consider an episode $\tau$ of length $m+1$ that takes place on a discrete time-grid of fixing times expressed in year fractions $\mathcal{T}\coloneqq \{T_0,...,T_k,...,T_m\}$ with $T_0\coloneqq 0$ and $T_m \coloneqq T$ maturity of the option. At a given episodic time $T_k$ the RL agent interacts with the environment: it receives a representation of the environment called state $s_k$ and on the basis of that it selects an action $a_k$ sampling from the current policy $\pi^{\theta_l}$. Here with $\theta_l$ we refer to the set of parameters through which we parameterize the agent policy at the $j$-th algorithm iteration. In our case the agent can choose the composition of the risky asset portfolio, so that the policy is the allocation strategy $\alpha$ introduced in \Cref{eq:TVS_elegant}:
\begin{equation}
 	a_k = \alpha_{T_k} \quad \forall T_k \in \mathcal{T}, \quad a_k \in \mathcal{A}\subset \mathbb{R}^n.
\label{eq:action}\end{equation} 

Since the value function of the problem depends on the Markovian state $X_t$, the portfolio level $I_t$ and time $t$, our natural choice for the observation state is the following block
\begin{equation}
 	s_{k} \coloneqq \left[X_{T_k}, I_{T_k}, T_k \right] \quad \forall T_k \in \mathcal{T}, \quad s_k \in \mathcal{S} \subset \mathbb{R}^{n+2}.
\label{eq:state}\end{equation}
In this way, the state contains all the information needed by the agent to take an optimal action, leading to the maximum plain vanilla TVO price. 

In this algorithm, we parameterize the agent policy with a feed forward neural network (FFNN), such that $\theta_l$ coincide with the hidden weights, $s_k$ with the inputs, and $a_k$ with the output. In this way, we are dealing with a deterministic policy, where its functional form is given by the neural network. The parameters update is performed as follows: the agent collects a finite batch of experiences interacting with the environment in a set of episodes $\tau$, then the loss function is evaluated as
\begin{equation}
L(\theta_l) = D(0,T;\zeta)\hat{\mathbb{E}}\left[ \Phi(I_T(\alpha(\theta_l))) \right],
\label{eq:direct_policy_loss}\end{equation}  
where the expectation $\hat{\mathbb{E}}[\dots]$ indicates the empirical average on the batch. Then the parameters are updated by plugging the policy into a stochastic gradient ascent algorithm. The choice of \Cref{eq:direct_policy_loss} is justified by the fact that, when the algorithm will find the optimal set $\theta^*$, then $L(\theta^*)$ will be a good proxy of the optimal option price. 

Once the training phase is ended and the agent has selected the optimal policy, we can run a Monte Carlo (MC) simulation with never seen scenarios to price the optimal target volatility option and test if the algorithm does not overfit the data.

\subsection{Proximal Policy Optimization Approach}
As for the direct policy approach, we model the pricing problem considering an episode that takes place on the time-grid $\mathcal{T}\coloneqq \{T_0,...,T_k,...,T_m\}$ with $T_0\coloneqq 0$ and $T_m \coloneqq T$. Again we choose as observation state the block defined in \Cref{eq:state} and the agent policy coincides with the risky allocation strategy $\alpha$ (\Cref{eq:action}). An important difference from the previous method is that: once the agent has selected the action $a_k$ sampled from the current policy, it receives at the next time $T_{k+1}$ a reward $r_{k+1}$ generated by the environment. 

In our work, we have defined two different reward functions with the purpose to analyze which could help the agent to learn more efficiently the optimal policy $\pi^* = \pi^{\theta^*}$. 
Our first definition is
 \begin{equation}
 	r_{k+1}=\left\{\begin{array}{ll}
 		(I_T(\alpha(\theta_l)) - K)^+ & \text { if } T_{k+1} = T \\
 		0 & \text { otherwise }
 	\end{array}\right.
 \label{eq:reward1}\end{equation}
Therefore, during the whole episode, the agent receives a nil reward except at maturity when the reward coincides with the option intrinsic value. This choice may seem too daring because the agent receives a real feedback of its actions only at the end of the whole episode, increasing the probability to obtain a slow learning. However, if the agent has learnt $\pi^*$, the average cumulative reward per episode will coincide with the optimal TVO price. 

The second reward function we have defined is 
\begin{equation}
 	r_{k+1} \coloneqq  \gamma^{k}[V_{\text{BS}}(T_{k+1})-V_\text{\text{BS}}(T_k)],
 \label{eq:reward2}\end{equation}
where $\gamma \in[0,1]$ is an hyper-parameter of the PPO, while $V_{\text{BS}}(T_k)$ is a proxy of the residual option price defined by
\begin{equation}
	V_{\text{BS}}(T_k) \coloneqq BS(F^{\text{TVS}}(T_k,T;\alpha^*_{\text{BS}},K),K,T-T_k,\bar{\sigma},D(T_k,T;\zeta)) \quad \text{and} \quad V_{\text{BS}}(T_0)=0,
\end{equation} 
with $\alpha^*_{BS}$ the BS optimal strategy \eqref{eq:BS_optimal_strategy} calculated in the state $s_k$. In this form the agent actions are hidden inside the term $I_t$ used to compute the TVS forward curve $F^{\text{TVS}}$ defined by \Cref{eq:tvs_forward}. In this case, the reward function does not suffer of nil values for $0<T_k<T$: the RL agent always gets a feedback from the environment for its choices.  The hyper-parameter $\gamma$ plays the role of a discount factor in the sense that, as $\gamma$ approaches to zero, the RL agent will tend to maximize immediate rewards while neglecting possible larger rewards in the future.
If we take the cumulative reward per episode and set the PPO parameters\footnote{We refer to \textcite{Schylman2016} for a more detailed description for the generalized advantage estimation parameter $\lambda$.} $\gamma =\lambda =1$ we obtain
\begin{equation}
	R(\tau) = \sum_{k=0}^{m-1} r_{k+1} \underset{\gamma=\lambda =1}{=} \sum_{k=0}^{m-1}[V_{\text{BS}}(T_{k+1})-V_\text{\text{BS}}(T_k)] = V_{BS}(T) = (I_T-K)^+,
\end{equation} 
which is equal to the intrinsic value of the option. This result does not depend on the definition of $V_{\text{BS}}$  $\forall T_k<T$, but we conjecture that the closer $V_{\text{BS}}$ is to the value function, the easier the agent is in learning.

Thus one can train the agent choosing the optimal value for $\gamma,\lambda\in[0,1]$, and then run, as test phase, a Monte Carlo simulation with $\gamma=\lambda=1$ and the optimized $\theta$ fixed, where, if the agent has learnt $\pi^*$, the average of $R(\tau$) along different episodes will match the optimal undiscounted price of the derivative contract on the TVS.

In the \textsc{OpenAI Baselines} implementation of the PPO, the agent policy is parameterized again by a neural network; as for the previous method we have chosen an FFNN. In particular with PPO we deal with a stochastic policy whose functional form is a multivariate diagonal Gaussian distribution where the mean $\mu^\theta(s_k)$ is the output vector of the FFNN and the log-standard deviation $\log\sigma$ is an external parameter
\begin{equation}
	\pi^\theta(s_k) \sim \mathcal{N}(\mu^\theta(s_k),e^{\log\sigma}).
\end{equation}
As one can observe, $\log\sigma$ is state-independent, but it is reduced as the number of the PPO update iterations increases. The idea is that the log-standard deviation will be higher at the beginning of the training phase in order to guarantee a good exploration of the action space while it will be lower at the end to avoid too much noise in the proximity of the optimal policy.

The fact that the PPO implementation exploits a stochastic policy ensures us a better exploration of the action space than with the previous approach. Moreover, the algorithm tries to learn an approximator of the on-policy value function as control variate for the training phase. This approximator is an FFNN with the same architecture as the one for the policy.

\section{Numerical Investigations}\label{sec:Numerical_investigations}
Here we present the numerical results obtained with our proposed methods. We focus our analysis on a European call option on a TVS with the following product details
\begin{equation}\nonumber
	I_0 = K = 1 \, \text{EUR}, \quad T=2 \, \text{yr}, \quad \bar{\sigma} = 5\%.
\end{equation}
Moreover without loss of generality we consider the case of a completely free allocation strategy $\alpha$. The extension to the constrained case is easy.

It is our aim to investigate the control problem under non-trivial dynamics like the local volatility one where the volatility of the risky asset is also a function of the state. By looking at the HJB equation \eqref{eq:HJB_general}, we expect that the volatility versor will play a role in finding the optimal solution, giving rise to a non-trivial strategy.

Although in \Cref{subsec:BS_model} we have proved that under the Black and Scholes model the \Cref{eq:BS_optimal_strategy} solves the control problem, we want to take advantage of this \textit{a priori} solution as a benchmark to gather evidence on the robustness of our Reinforcement Learning approach and to check if our analytical result is correct. Moreover, we use the BS model as numerical laboratory to perform fine-tuning tests for the RL algorithms hyper-parameters and analyze how they impact the final results and performances.

We recall that in both the algorithms we parameterize the agent policy with an FFNN; thus this is completely characterized by the following hyper-parameters: number of hidden layers, number of neurons per hidden layer, activation function per hidden layer. This is due to the fact that in this Reinforcement Learning problem the number of the input neurons is equal to the state space dimension \eqref{eq:state}, while the number of the output ones coincides with the action space dimension \eqref{eq:action}. It is well known in the literature that neural networks give better performances in the training phase if the input data are well normalized \citep{Sola1997,Puheim2014}. Thus we choose as state $s_k$ the following normalized block  

\begin{equation}
	[\log(S_{T_k}/F(0,T_k)), I_{T_k}/I_0, T_k ] \quad \forall T_k\in\mathcal{T},
\label{eq:input_state}\end{equation}
where $F(t,T)$ is the forward curve vector of the risky assets from $t$ to $T$. In \Cref{eq:input_state} we have chosen as Markovian state $X_{T_k}$ the martingale term of the securities dynamics. In this way we have that in the input block the variables have the same order of magnitude.

\subsection{Black and Scholes: Hyper-Parameters Fine Tuning}
We use the BS environment as toy model to understand which parameters of the RL algorithm play key roles in the training and testing phase. 
Our first approach has been the direct policy one since it represents a natural way to tackle the problem: since our goal is to find the allocation 
strategy that maximizes the option price, we update the FFNN weights following the gradient direction of the loss function defined in \Cref{eq:direct_policy_loss}. 

We try to investigate the following hyper-parameters: the FFNN architecture, in particular which feature between the depth and the width of the network is more important, the activation function and the learning rate of the optimizer. Moreover, we compare the performances of two well-known optimizers in Machine Learning literature: Nadam \citep{Dozat2016} and RMSprop \citep{Hiton2012}. Since we deal with a free allocation strategy $\alpha$ that can assume negative values, we have chosen among the wide variety of activation functions the tanh and the elu. In this way, we can analyze the performance of a saturating activation function and a non-saturating one. Firstly we have performed a grid search on the learning rate starting value and we have found $10^{-3}$ a good choice in terms of speed of learning and avoiding over-fitting.
 \begin{figure}[tb]
	\centering
	\includegraphics[width=7.2cm,height=6cm]{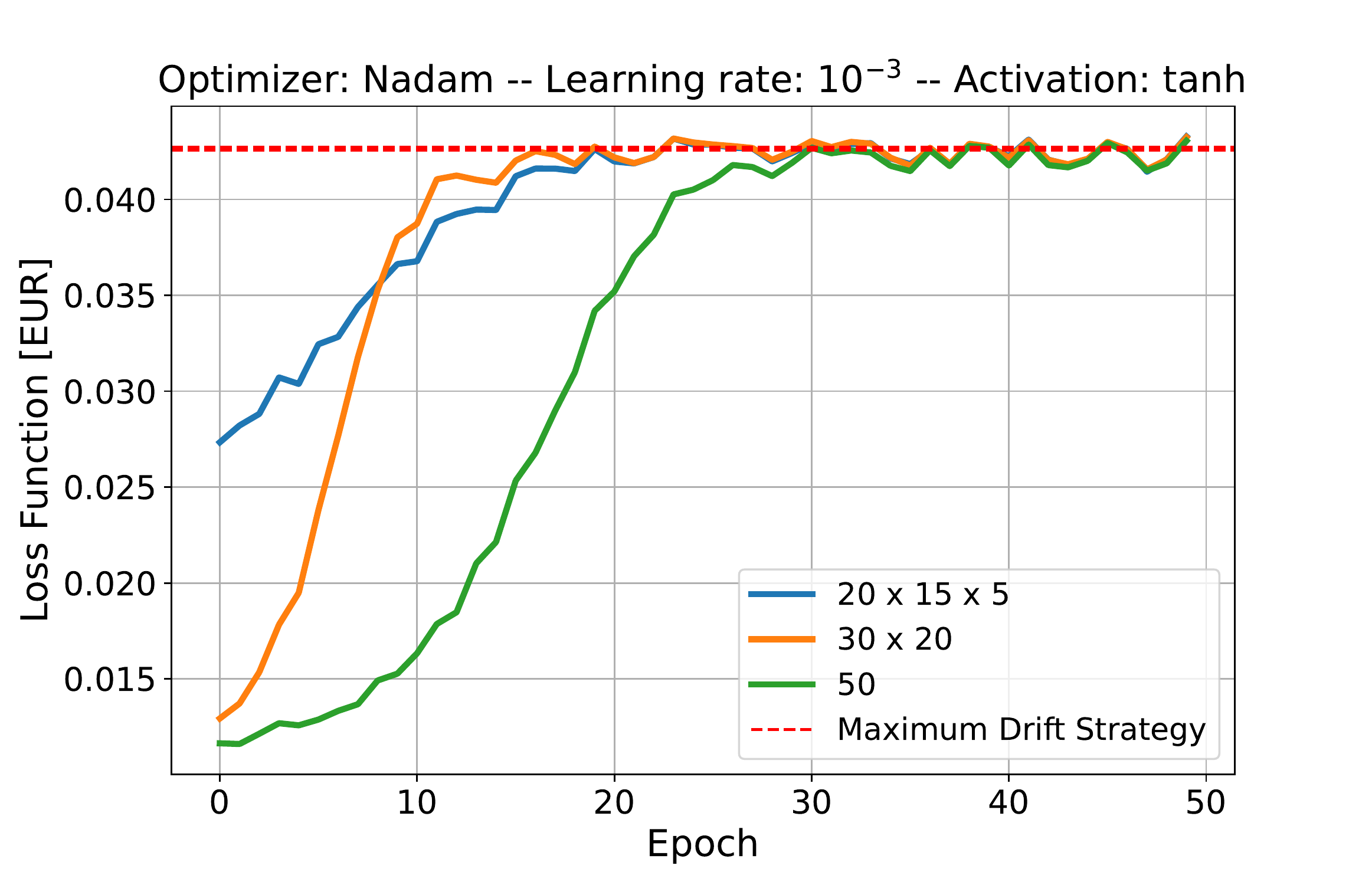}
	\includegraphics[width=7.2cm,height=6cm]{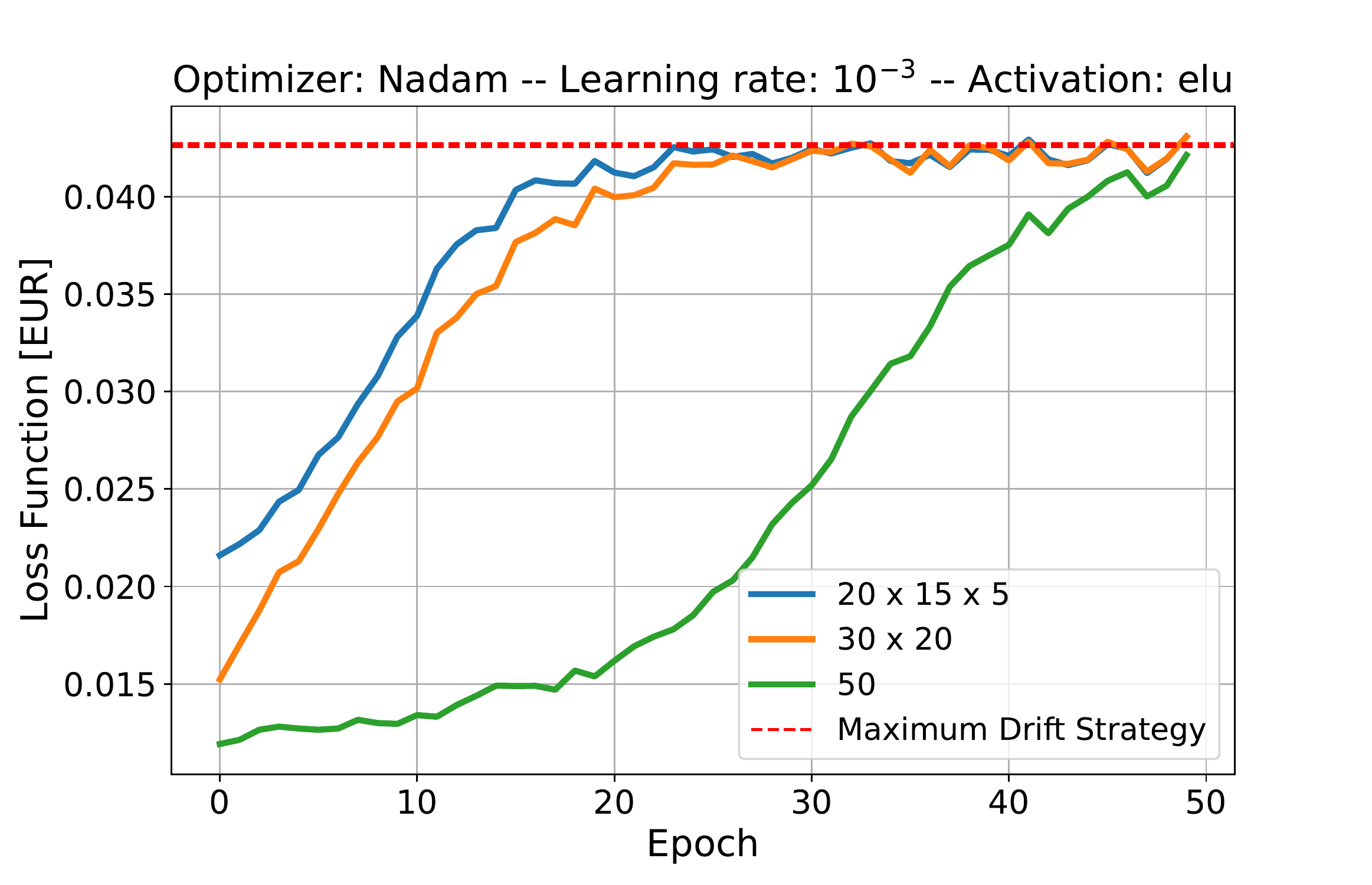}
	\caption{Learning curves of the direct policy algorithm applied to the Black and Scholes model. The solid lines are the loss functions defined in \Cref{eq:direct_policy_loss} in function of the number of training epochs. Each continuous line represents a different neural network architecture with tanh activation function (left) and elu activation function (right). The optimizer adopted is the Nadam with a learning rate of $10^{-3}$. The horizontal red-dashed line is the conservative option price obtained by maximizing the TVS drift through \Cref{eq:analytical_solution_free_strategy}.}
	\label{fig:result_bs_nadam}
\end{figure}  
 \begin{figure}[tb]
	\centering
	\includegraphics[width=7.2cm,height=6cm]{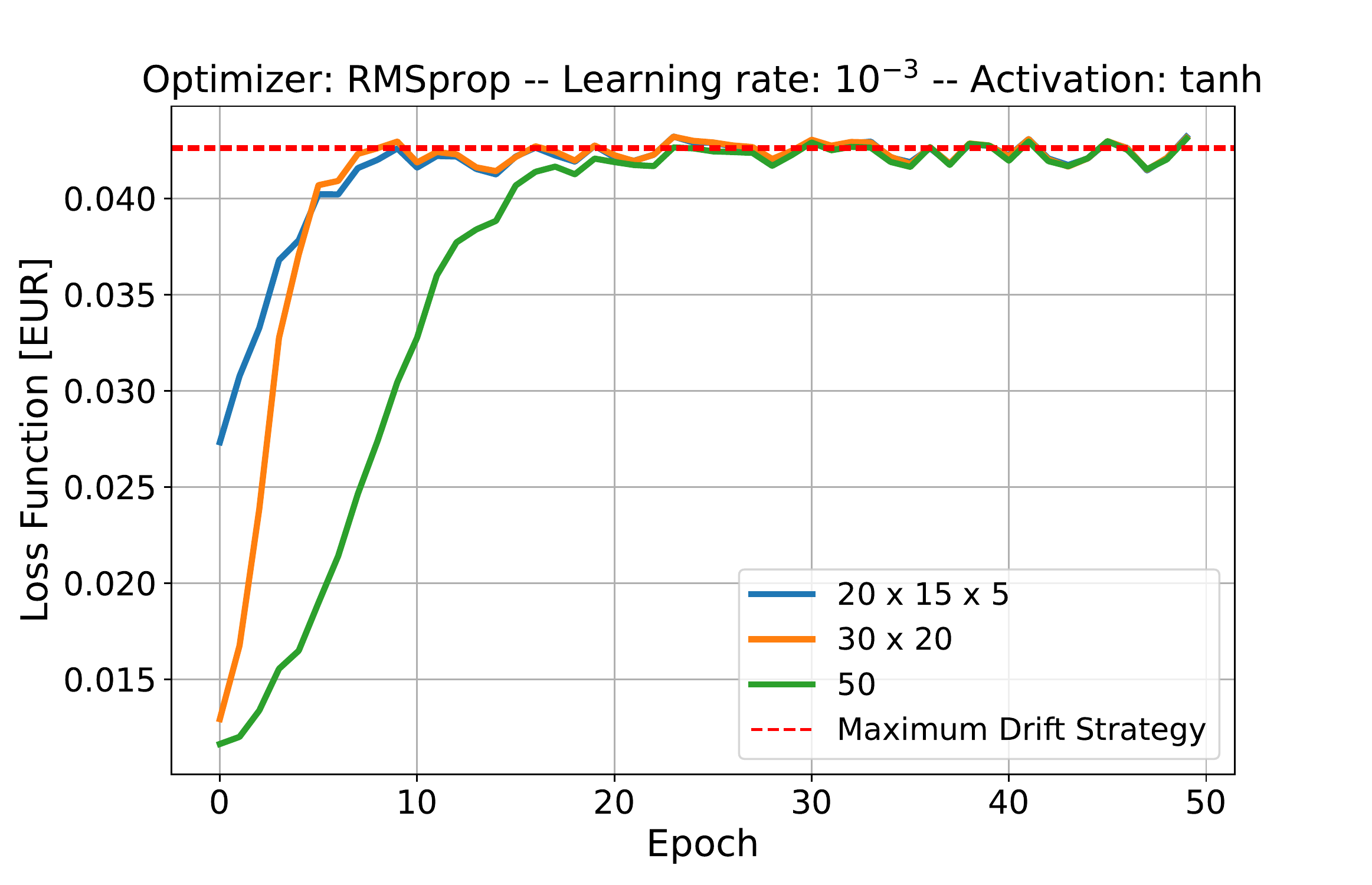}
	\includegraphics[width=7.2cm,height=6cm]{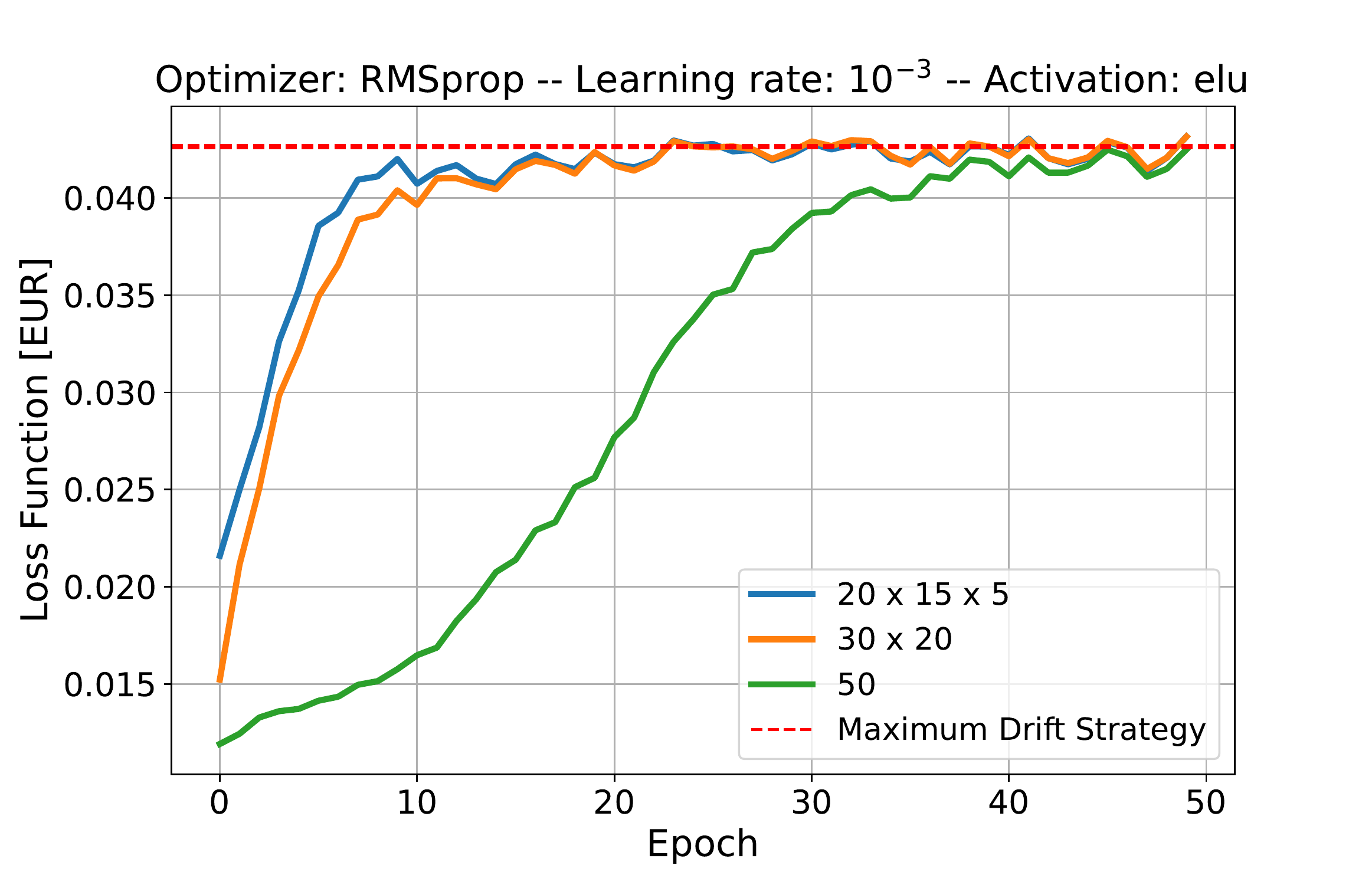}
	\caption{Learning curves of the direct policy algorithm applied to the Black and Scholes model. The solid lines are the loss functions defined in \Cref{eq:direct_policy_loss} in function of the number of training epochs. Each continuous line represents a different neural network architecture with tanh activation function (left) and elu activation function (right). The optimizer adopted is the RMSprop with a learning rate of $10^{-3}$. The horizontal red-dashed line is the conservative option price obtained by maximizing the TVS drift through \Cref{eq:analytical_solution_free_strategy}.}
	\label{fig:result_bs_rmsprop}
\end{figure}  

In \Cref{fig:result_bs_rmsprop,fig:result_bs_nadam} we present the learning curves of our fine-tuning tests for the RMSProp and the Nadam respectively. The three lines of each plot display the learning curves for different FFNN architectures: a one hidden layer network with 50 neurons, a 2 hidden layers with 30 and 20 neurons respectively, and a deeper one with three layers with 20, 15, and 5 neurons. Each learning curve is compared with the optimal option price (red dotted line) that we have computed with the closed form solution we have derived for the BS model in \Cref{subsec:BS_model}. We can see that all the learning curves converge to our expected price, providing a numerical demonstration of our theoretical result. More deeply, from the plots we observe that the RMSprop optimizer outperforms the Nadam. Moreover, the tanh activation function seems to be more preferable than the non-vanishing elu. However, more importantly, we have evidence that a deeper architecture of the neural network outperforms the shallow one. All the learning curves we have presented are the best-in-sample results, in terms of performance, of four runs with different random starting guesses for the hidden weights $\theta$. This procedure is necessary since the objective function is not convex.

We take the $20\times 15 \times 5$ network with tanh from RMSprop as the best optimized network, and we run a Monte Carlo simulation with $10^6$ never-seen scenarios to check if the agent overfits the new data. 
\begin{table}[tb]
	\centering
	\begin{tabular}{|c|c|} \hline
		Method & TVO price [EUR]  \\ \hline 
		Analytical Solution & $4.2634\times 10^{-2}$ \\
		Direct policy RL & $(4.2624 \pm 0.005) \times 10^{-2}$ \\ \hline
	\end{tabular}\caption{Comparison of TVO prices under Black and Scholes model: analytical solution price and direct policy reinforcement learning. The option parameters are: $I_0=K=1\,[\text{EUR}]$, $T=2\,[\text{yr}]$ and $\bar{\sigma}=5\%$.} 
\label{tab:mc_results_bs}\end{table}
We report the results in \Cref{tab:mc_results_bs}: the RL price is compatible with the closed formula price. Thus the RL agent did not overfit the data during the training phase.

We will take advantage of those fine-tuning results to tackle the local volatility problem.

\subsection{Local Volatility Dynamics}
In this section, we study the TVS control problem assuming a local volatility model for the dynamics of the risky assets. Thus in this case we have a diffusive term in \Cref{eq:Equity_process} that is a deterministic function both of time and state, i.e. the spot price, $\nu_t=\nu(t,S_t)$. This additional dependency of volatility makes the problem of finding the optimal strategy non-trivial; in fact if we consider the whole securities smiles then the second-order term in the HJB equation \eqref{eq:HJB_general} is not zero and thus a closed formula for $\alpha^*$ is not available anymore. Because of that, one must resort to numerical techniques to recover the problem solution. Unlike the BS model where the strategy depends only on time, in LV dynamics there is no unnecessary information provided by the state block in \Cref{eq:state} to take the optimal action. Here we consider the same market data ($r_t$, $\mu_t$ and $\nu_t$) as in the Black and Scholes environment to study how the optimal solution changes with the dynamics model. 

\begin{figure}[tb]
	\centering
	\includegraphics[width=9cm,height=6cm]{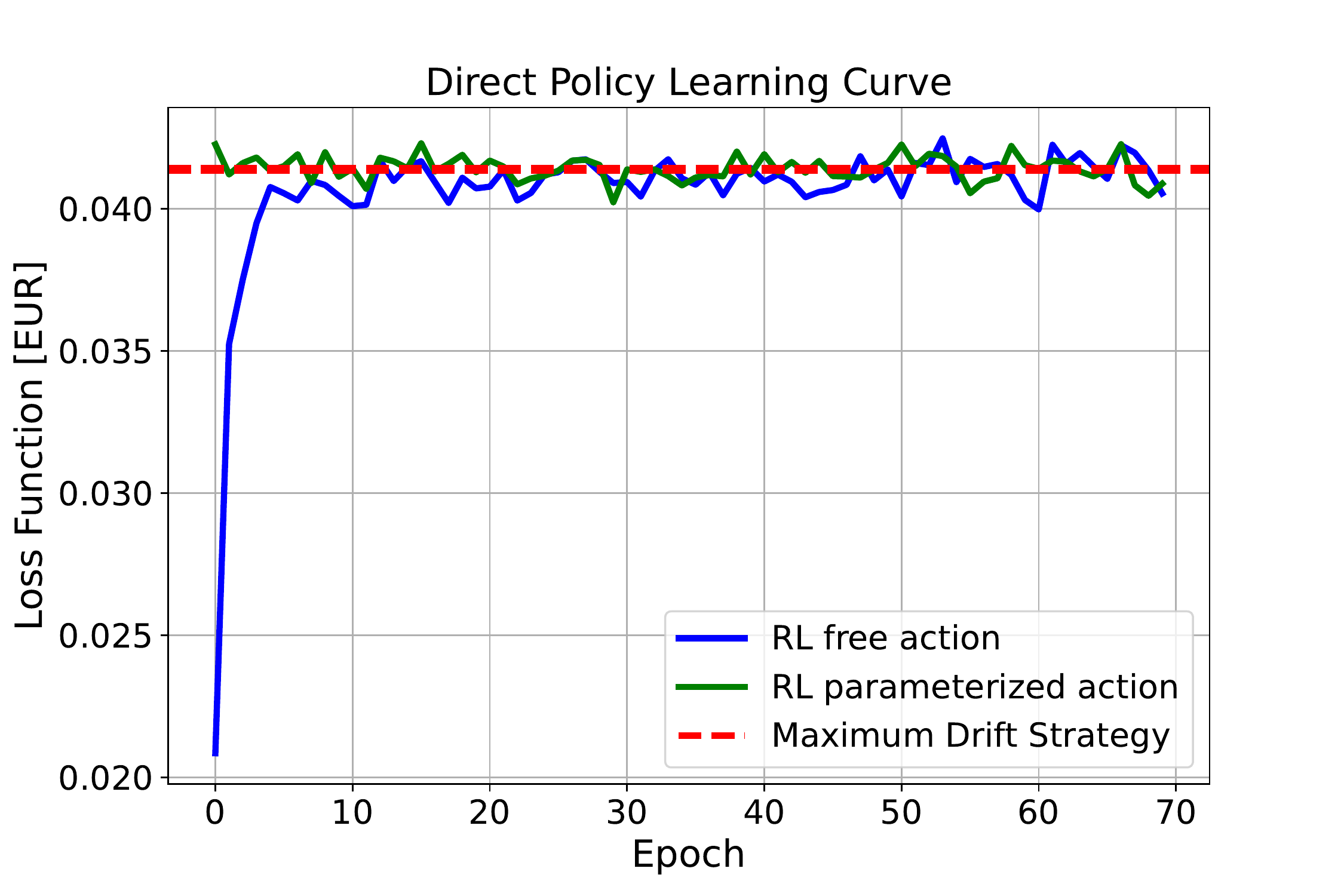}
	\caption{Learning curves of the direct policy algorithm applied to the local volatility model. The solid blue line is the learning curve of an agent with free actions, while the green solid one is the learning curve of an agent whose actions are parameterized with the maximum drift strategy. The two agents are parameterized by a $20 \times 10 \times 5$ FFNN with tanh activation function. The loss function is computed according to \Cref{eq:direct_policy_loss}. The optimizer adopted is the RMSprop with a learning rate of $10^{-3}$ for the blue line while $10^{-4}$ for the green one. The horizontal red dashed lines are the 99\% confidence interval of the MC price obtained by maximizing path-wise the TVS drift through \Cref{eq:analytical_solution_free_strategy}.}
	\label{fig:result_lv_direct_policy}
\end{figure}  
Our first way to tackle the problem is to exploit the results obtained in the BS case. Thus we train by direct policy algorithm a deep $20\times 10\times 5$ neural network with tanh activation function and a RMSprop optimizer. We report in \Cref{fig:result_lv_direct_policy} the corresponding learning curve (blue line). From this, we observe that the network has learnt some good policy since the curve grows as the number of training epochs increases until it saturates at a certain value. To try to measure the performance of the policy selected by the agent, we can build a n\"aif strategy called \virgolette{baseline}. By looking at the Markovian projection in \eqref{eq:markovian_projection}-\eqref{eq:local_drift} and the first term in the HJB \eqref{eq:HJB_general}, a natural choice for the baseline is to maximize the TVS local drift. In other words this coincides by applying the path-wise the Black and Scholes solution of \Cref{eq:BS_optimal_strategy}; since here we deal with a free allocation strategy, we can simply use our analytical result \eqref{eq:analytical_solution_free_strategy}. In the same \Cref{fig:result_lv_direct_policy} we report the 99\% confidence interval of the MC price obtained with the maximum drift baseline as two red-dashed lines. As we can see, the optimized loss function is compatible with the baseline price. Following the theoretical result of the HJB equation, we can assert that the agent has learnt a sub-optimal policy. Since we LV model differs from the BS one for a corrective term in the HJB \eqref{eq:HJB_general}, we expect that the optimal solution  in the LV framework will be in a close region of the maximum drift strategy. Thus we train another agent with the direct policy learning by parameterizing its action with the baseline strategy and choosing a smaller learning rate of $10^{-4}$. With this parameterization, at each observational time $T_k\in\mathcal{T}$ the risky allocation strategy is obtained by summing the network output with the \Cref{eq:analytical_solution_free_strategy}. Even in this case, the corresponding learning curve (green line in \Cref{fig:result_lv_direct_policy}) is stuck in the maximum drift strategy. The first possible reason for this behaviour is that the learning strategy of the direct policy approach is too simple for the learning task. The second interpretation is that the volatility versor in the HJB does not affect significantly the optimum location since it is a second-order term and it is lost in the MC error. 

Because of that, we change the learning strategy by adopting a more sophisticated one: the PPO algorithm. The big advantage of PPO is that, in addition to the policy, a guess of the value function in \Cref{eq:HJB_general} is also computed by a parameterization through an FFNN. As for the direct policy approach, we use the BS environment to fine-tune the PPO hyper-parameters. Again we experienced that deeper FFNNs outperform shallow ones and tanh is the most preferable. We report for completeness the values of the other PPO hyper-parameters, for the description of which we refer to \textcite{Schulman2017}: learning rate $3\times 10^{-4}$, $\lambda=0.95$, $\epsilon=0.2$, $c_1=0.7$, $c_2=0$ and mini-batch size of 2048 episodes. 

We have trained a 5 layer FFNN with 8 neurons each with the PPO algorithm in two different environments: one environment generates the reward according to \Cref{eq:reward1}, while the other exploits the reward function \eqref{eq:reward2} where we set $\gamma=0.98$ as discount factor to make the agent prefer immediate rewards. Again, for both environments, we train two different agents: one whose action is given by \Cref{eq:action}, while the other implements the action parameterization with the baseline.
\begin{figure}[tb]
	\centering
	\includegraphics[height=6cm,width=7.2cm]{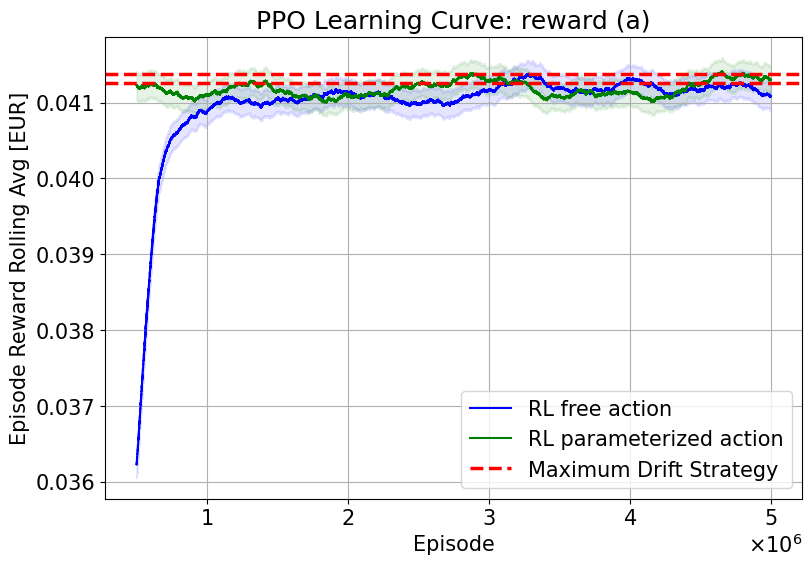}
	\includegraphics[height=6cm,width=7.2cm]{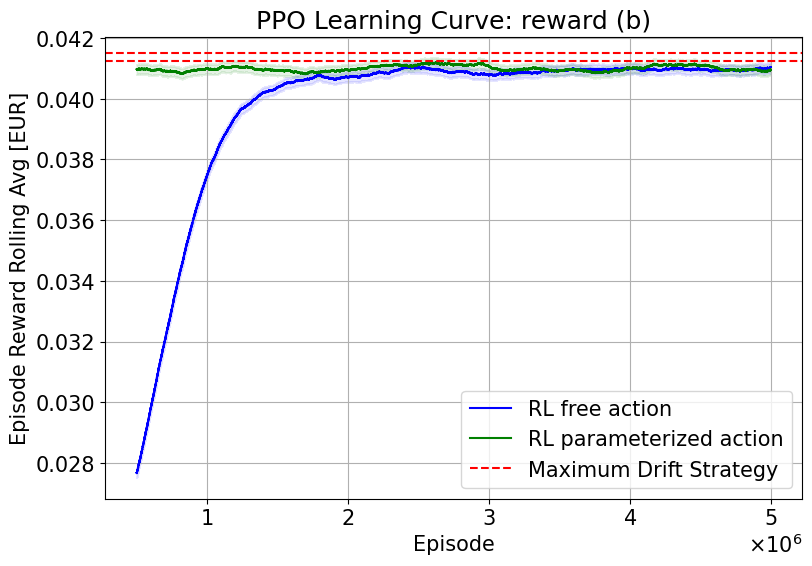}
	\caption{Learning curves of the PPO algorithm applied to the local volatility model. On the horizontal axis the number of training episodes. The solid lines are the moving average of the realized rewards on the last $10^5$ episodes. The shadows represent the $98\%$ confidence intervals. The label reward (a) indicates that the RL agent is trained in an environment where the reward function is defined by \Cref{eq:reward1}, while reward (b) defines the reward function in \Cref{eq:reward2}. The solid blue lines are the learning curves of an agent with free actions, while the green solid ones are the learning curves of an agent whose actions are parameterized with the maximum drift strategy. In both plots, the agents are parameterized by an FFNN with 5 hidden layers with 5 neurons each with tanh activation function. The horizontal red dashed lines are the 99\% confidence interval of the MC price obtained by maximizing path-wise the TVS drift through \Cref{eq:analytical_solution_free_strategy}.}\label{fig:PPO_results}
\end{figure}  
We report in \Cref{fig:PPO_results} the results of the learning curves. In all the PPO cases, the trained agents seem to be stuck in the sub-optimal maximum drift strategy policy. For the case of the immediate reward function (\Cref{fig:PPO_results} on the right), we have that the learning curves converge to a saturation value that is not compatible with the baseline price; this is due to the fact that with $\lambda,\gamma \neq 1$ we introduce a bias in the problem. Because of that, we need to run a test MC simulation with freezed agents and $\lambda=\gamma=1$ to obtain a result to compare with the baseline price. We report the results in \Cref{tab:mc_results_lv}. The Reinforcement Learning agents we have trained with PPO give all prices compatible with the baseline one.
\begin{table}[tb]
	\centering
	\begin{tabular}{|c|c|} \hline
		Method & TVO price [EUR]  \\ \hline 
		Baseline &  $(4.138 \pm 0.005)\times 10^{-2}$\\
		RL free (a) & $(4.127 \pm 0.005) \times 10^{-2}$ \\
		RL parameterized (a) & $(4.131 \pm 0.005) \times 10^{-2}$ \\
		RL free (b) & $(4.130 \pm 0.005) \times 10^{-2}$ \\
		RL parameterized (b) & $(4.135 \pm 0.005) \times 10^{-2}$ \\ \hline
	\end{tabular}\caption{Comparison of TVO prices under local volatility model. The baseline price is obtained by applying path-wise the maximum drift strategy \eqref{eq:analytical_solution_free_strategy}. The RL free agent implements the policy defined by \Cref{eq:action}, while the parameterized agent chooses its actions in terms of the baseline strategy. We label with letter (a) the agents trained in environments with the reward function \eqref{eq:reward1}, while with (b) the reward function \eqref{eq:reward2}. The option parameters are: $I_0=K=1\,[\text{EUR}]$, $T=2\,[\text{yr}]$ and $\bar{\sigma}=5\%$.}
	\label{tab:mc_results_lv}\end{table}

\section{Conclusion and Further Developments}
In this paper, we described a non-trivial control problem related to derivative contracts on target volatility strategies. In particular, we have considered a bank selling a call option to a fund manager as protection on the capital invested on the TVS. We showed how the presence of different funding costs coming from hedging the risky assets underlying the TVS, obliges the bank to solve a stochastic optimal control problem to price the protection. This is due to the fact that the bank strategy is not self-financing. This kind of control problem is hard to solve because here the control process affects both drift and diffusive coefficients of the controlled process. Despite its complexity, our first contribution is the derivation of a closed form solution of the control problem in a Black and Scholes framework, which could represent a useful tool for practitioners since it outperforms intuitive trading strategies. We have derived this solution in two different ways: first by applying the Gy\"ongy  Lemma and then by writing the Hamilton-Jacobi-Bellman equation.  We numerically studied the problem in the more general local volatility model where the solution is not available and thus an numerical investigation is needed. We tackled the problem by means of the novel Reinforcement Learning techniques, by both the direct policy learning and the proximal policy optimization one. We used the BS model, where the solution is \textit{a priori} known, as benchmark to perform a series of fine-tuning of the RL algorithm hyper-parameters, such as the artificial neural network architecture. We have tested in the LV model the two RL approaches and from our simulations we have evidence that nor the simple direct policy learning strategy nor the sophisticated PPO are able to outperform our analytical solution applied path-wise. Thus our analytical result for the Black and Scholes model seems to be a good proxy solution also for the local volatility one.

 This result seems to be a local optimum from the HJB equation of the problem, since in the LV model the volatility versor term should influence the RL agent actions. Thus natural development of this work could be to solve the HJB numerically in low dimension in order to check why such sophisticated algorithms are not able to find the global optimum of the problem, or to understand which are the key elements of the problem, such as market data or the payoff function, that can give rise to a solution far from the intuitive one.

\printbibliography
\end{document}